\documentclass[letterpaper, 10 pt, conference]{ieeeconf}
\usepackage{amsthm}
\usepackage[english]{babel}

\IEEEoverridecommandlockouts
\usepackage{cite}
\usepackage{algpseudocode}  

\usepackage{amsmath,amssymb,amsfonts}
\usepackage{graphicx}
\usepackage{textcomp}
\usepackage{xcolor}
\usepackage{colortbl}
\usepackage[caption=false]{subfig}
\usepackage{multirow}
\usepackage{booktabs}
\usepackage[colorlinks,linkcolor=blue]{hyperref}
\usepackage{siunitx}

\usepackage{threeparttable}
\usepackage{algorithm}
\newcommand{\T}{\mathsf{T}}

\def\calX{\mathcal{X}}
\def\calU{\mathcal{U}}

\newtheorem{theorem}{Theorem}
\newtheorem{lemma}{\bf Lemma}

\def\BibTeX{{\rm B\kern-.05em{\sc i\kern-.025em b}\kern-.08em
    T\kern-.1667em\lower.7ex\hbox{E}\kern-.125emX}}
\begin{document}

\title{\LARGE \bf Canonical Form of Datatic Description in Control Systems}

\author{Guojian Zhan, Ziang Zheng, Shengbo Eben Li
\thanks{*This work is supported by NSF China with U20A20334 and Tsinghua University Initiative Scientific Research Program. It is also supported by Tsinghua University-Toyota Joint Research Center for AI Technology of Automated Vehicle. G. Zhan and Z. Zheng contributed equally. All correspondences should be sent to S. E. Li with email: {\tt\small lisb04@gmail.com}.}
\thanks{G. Zhan, Z. Zheng and S. E. Li are with School of Vehicle and Mobility, Tsinghua University, Beijing, 100084, China.}%
}

\maketitle
\begin{abstract}
The design of feedback controllers is undergoing a paradigm shift from modelic (i.e., model-driven) control to datatic (i.e., data-driven) control. 
Canonical form of state space model is an important concept in modelic control systems, exemplified by Jordan form, controllable form and  observable form, whose purpose is to facilitate system analysis and controller synthesis. 
In the realm of datatic control, there is a notable absence in the standardization of data-based system representation. 
This paper for the first time introduces the concept of \textit{canonical data form} for the purpose of achieving more effective design of datatic controllers. In a control system, the data sample in canonical form consists of a \textit{transition} component and an \textit{attribute} component. The former encapsulates the plant dynamics at the sampling time independently, which is a tuple containing three elements: a state, an action and their corresponding next state. The latter describes one or some artificial characteristics of the current sample, whose calculation must be performed in an online manner. The attribute of each sample must adhere to two requirements: (1) causality, ensuring independence from any future samples; and (2) locality, allowing dependence on historical samples but constrained to a finite neighboring set. 
The purpose of adding attribute is to offer some kinds of benefits for controller design in terms of effectiveness and efficiency. 
To provide a more close-up illustration, we present two canonical data forms: temporal form and spatial form, and demonstrate their advantages in reducing instability and enhancing training efficiency in two datatic control systems.

\end{abstract}

\section{Introduction}

The development of control theory has undergone a transformative journey, evolving from classical methods centered around transfer functions in the frequency domain to modern methods relying on state space models in the temporal domain \cite{aastrom2014control}. Modern control theory leverages linear algebra techniques to simplify tasks such as system modeling, structure transformation, modal analysis, and controller synthesis. With developments spanning several decades, there have been many model-based methods of designing feedback controllers including linear quadratic control, H-infinity control, and model predictive control \cite{li2017robust, li2010model}. 
Since World War II, the advent of digital computers has propelled substantial progress in the industrial application of modern control methods, particularly in fields such as satellite navigation, rocket control, and autonomous driving \cite{canuto2008drag, perez2020model, zhan2024transformation}.

The canonical form of state space model is an unavoidable concept in modern control methods, which provides great benefits in understanding system characteristics and designing effective controllers \cite{kalman1962canonical}. The underlying mathematical principle lies in the characteristic invariance of state space representation under the nonsingular transformation on system states \cite{ho1980team}. Consider a linear time-invariant plant described by $\dot{x} = Ax + Bu$, where $x\in \mathbb{R}^n$ and $u\in \mathbb{R}^m$ represent the state and action, and $A$ and $B$ are the system and input matrices, respectively. Through a nonsingular transformation matrix $P\in\mathbb{R}^{n\times n}$ for the transformation $\hat{x}=Px$, the resulting state space model $\dot{\hat{x}} = P^{-1}AP\hat{x} + P^{-1}Bu$ encapsulates the same plant dynamics.
In fact, different canonical forms utilize distinct transformation matrices to yield diverse state space models to facilitate system analysis and controller synthesis \cite{kalman1982computation}. For example, representing a state space model in controllable canonical form or observable canonical form enables a direct assurance of controllability or observability without additional judgment procedures \cite{KH}. In addition, any state space model can be transformed into the Jordan canonical form, wherein the system matrix becomes a diagonal square matrix, and the diagonal elements represent the eigenvalues of the plant \cite{pratzel1983canonical}. Therefore, the stability verification becomes straightforward by examining whether all eigenvalues are negative. Furthermore, different eigenvalues indicate different modalities of the plant, enabling a more targeted design of the controller. 

As plant dynamics become more complex in recent years, the design of feedback controllers is undergoing a paradigm shift from modelic (i.e., model-driven) control to datatic (i.e., data-driven) control. Existing control methods, including both classical and modern versions, are encountering significant challenges due to sophisticated system behaviours and dynamic operating environments. The primary hurdle arises from the difficulty in constructing an accurate yet structurally simple system model. Even when an inaccurate model is attainable, it often assumes a highly intricate mathematical form, which introduces heavy computational burden in system analysis and subsequent controller design.

While modeling a system has become progressively intricate, collecting its behaviour data has considerably simplified due to rapid advance in storage and communication technologies. In recent years, there has been a notable surge in algorithms design for datatic controllers, where the term ``datatic'' underscores an emphasis on solemnly utilizing data samples in the design process. 
One prominent topic in this domain is reinforcement learning \cite{duan2021distributional}: it gathers samples through iterative interaction with the environment, forming a data-driven representation of environmental dynamics. With sufficient data, an optimal policy can be trained using policy iteration or value iteration, ensuring optimality based on the Bellman equation \cite{li2023rlbook}. This trained policy can subsequently be applied online for closed-loop control.
Other examples include data-driven predictive control, iterative linear quadratic control \cite{berberich2020data, fridovich2020efficient}. A crucial distinction between these datatic methods and the previously mentioned modelic methods is their direct dependence on interaction data to describe the plant dynamics, instead of relying on a fitted model obtained through system identification. 

Since datatic control is centered around data, how to collect and store data has a major influence on the effectiveness and efficiency of designing a controller. It is crucial to emphasize the importance of efficiency, particularly in complex systems such as autonomous vehicles and humanoid robotics, where training a useful controller from data often requires considerable time, ranging from tens of hours to days \cite{jiang2023reinforcement, feng2023dense, wurman2022outracing, ju2022transferring}. Consequently, even a modest enhancement in efficiency can result in substantial cost savings and labor force reduction. Therefore, a question naturally arises that whether there could be more effective forms for data representation, which is analogous to the canonical forms of state space models. This inspires us to explore the canonial form of datatic representation, a topic that remains both practically valuable and theoretically important.

This paper for the first time introduces the concept of \textit{canonical data form} into datatic control systems. The canonical data form establishes a standardized framework for datatic representation with the aim of providing benefits to controller design in terms of effectiveness and efficiency.
Our main contributions are summarized as follows.

\begin{enumerate}
    \item The canonical data form is developed for more effective datatic representation of system
    dynamics. Specifically, a data sample in canonical form consists of a \textit{transition} component and an \textit{attribute} component. The former encapsulates the plant dynamics at the sampling time independently. The latter describes one or some artificial characteristics of the current sample.  
    In our framework, different canonical forms can be customized according to specific needs to facilitate the development of datatic controllers.
    \item Two canonical data forms, i.e., temporal form and spatial form, are presented to offer a more close-up illustration. 
    The former uses the consumed time between  two chosen events as the temporal attribute. One can specify the minimum time cost for for each pair of events and fit an event-time distribution to serve as an additional performance measure.
    The latter needs to select a few fixed anchors in the state space, and uses the distances to these anchors as the spatial attribute. One can leverage these distances to search an arbitrary sample in a much faster way. Experiments are conducted to demonstrate their benefits in reducing instability and enhancing efficiency in  datatic controller design.
\end{enumerate}
\section{Datatic control system}

The control paradigm can be broadly categorized into two groups based on how to describe the system dynamics: (1) modelic control and (2) datatic control, as illustrated in Fig. \ref{fig.datatic_and_modelic_comparison} \cite{yang2024stability}. The term ``modelic'' means being driven by or based on or related to models, either transfer function or state space model. The term  ``datatic'' has a corresponding meaning related to data, which is collected by interaction with environment. Collecting data also happens in modelic control paradigm. But in this paradigm, data is utilized to fit a parameterized model through system identification, and  controllers are still synthesized based on models. In contrast, a datatic controller is  directly solved using data, omitting the step of building a model with system identification.

\begin{figure}[!htbp]
\centering
\includegraphics[width=0.48\textwidth,trim={0cm 0.8cm 0cm 0.5cm}, clip]{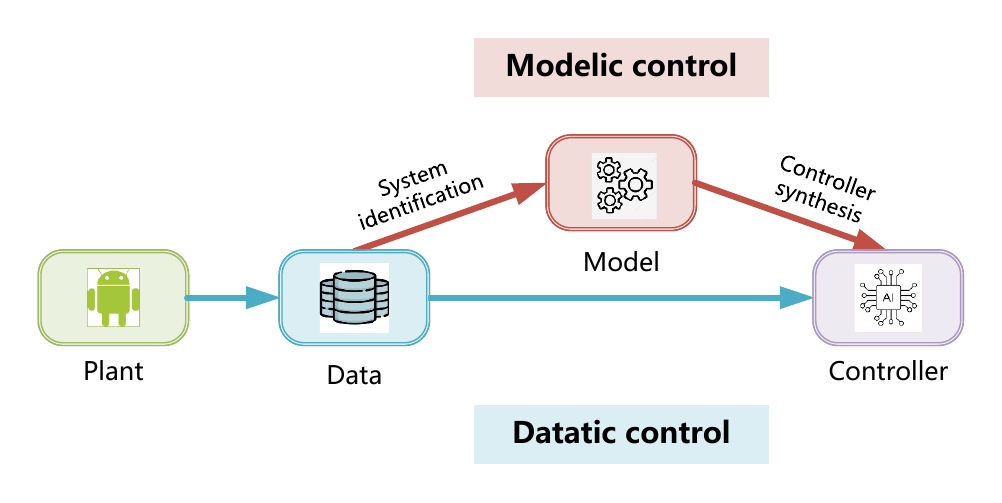}
\caption{Two types of control paradigms. Modelic control (on the upper
path) first fit a model with system identification and then use this model to synthesize controllers.
Datatic control (on the lower path) directly solves controllers using data.}
\label{fig.datatic_and_modelic_comparison}
\end{figure}

Both modelic control and datatic control have their own advantages and disadvantages due to their different system behaviour representations. 
In a modelic control system, it needs to fit the system model with a function of a specific form. The model provides a continuous description of system dynamics, meaning it can generate an output at every point in the state-action space. However, modelic description is susceptible to errors because the true system may not precisely match the assumed function form. In a datatic control system, explicit models are not constructed; instead, data samples are directly utilized to describe the system dynamics. This form of system behavior representation is termed as datatic description. If data samples are abundant enough, they can offer an accurate portrayal of system dynamics, at least at their respective locations, as they originate from the direct measurement of system states. One might question the accuracy of datatic description in the presence of perception errors. Actually, sensors can be considered as part of a closed-loop system, and thus, their errors contribute to the overall system dynamics. Consequently, sensor measurements occur discretely, both temporally and spatially, represented by a limited number of data points. No information is available within the intervals between these data points. The datatic information of a dynamic system must be discrete rather than continuous in the state-action space.

A standard datatic control system includes a set of input and output data collected by interacting with a plant, which is denoted as
\begin{equation}
\label{eq: datatic_control_system}
    \mathcal{D}=\{(x,u,x')_i|1\le i\le N\},
\end{equation}
where $x \in\calX\subseteq\mathbb{R}^n$ is a state, $u \in\calU\subseteq\mathbb{R}^m$ is an action, $x'$ is the next state obtained by applying $u$ on the plant at $x$, and $N$ is the number of data samples.
The dataset $\mathcal{D}$ is a datatic description of a discrete-time plant
\begin{equation}
\label{eq: datatic_dynamics}
    x'=f(x,u).
\end{equation}
where $f:\calX\times\calU\to\calX$ is an unknown system dynamics. The function $f$ is a modelic description of plant
dynamics. If it is known and accurate, we can use it to design controllers. Unfortunately, models are usually inaccurate
or even unknown in many real-world tasks. How to efficiently
use data to design controllers is a central task of datatic control paradigm.
\section{Canonical data form}
This section will first present the definition of canonical data form. Then the temporal form and spatial form will be introduced in detail.

\subsection{Definition of canonical data form}

\begin{figure}[!htbp]
\centering
\includegraphics[width=0.49\textwidth,trim={0cm 0.2cm 0cm 0.0cm}, clip]{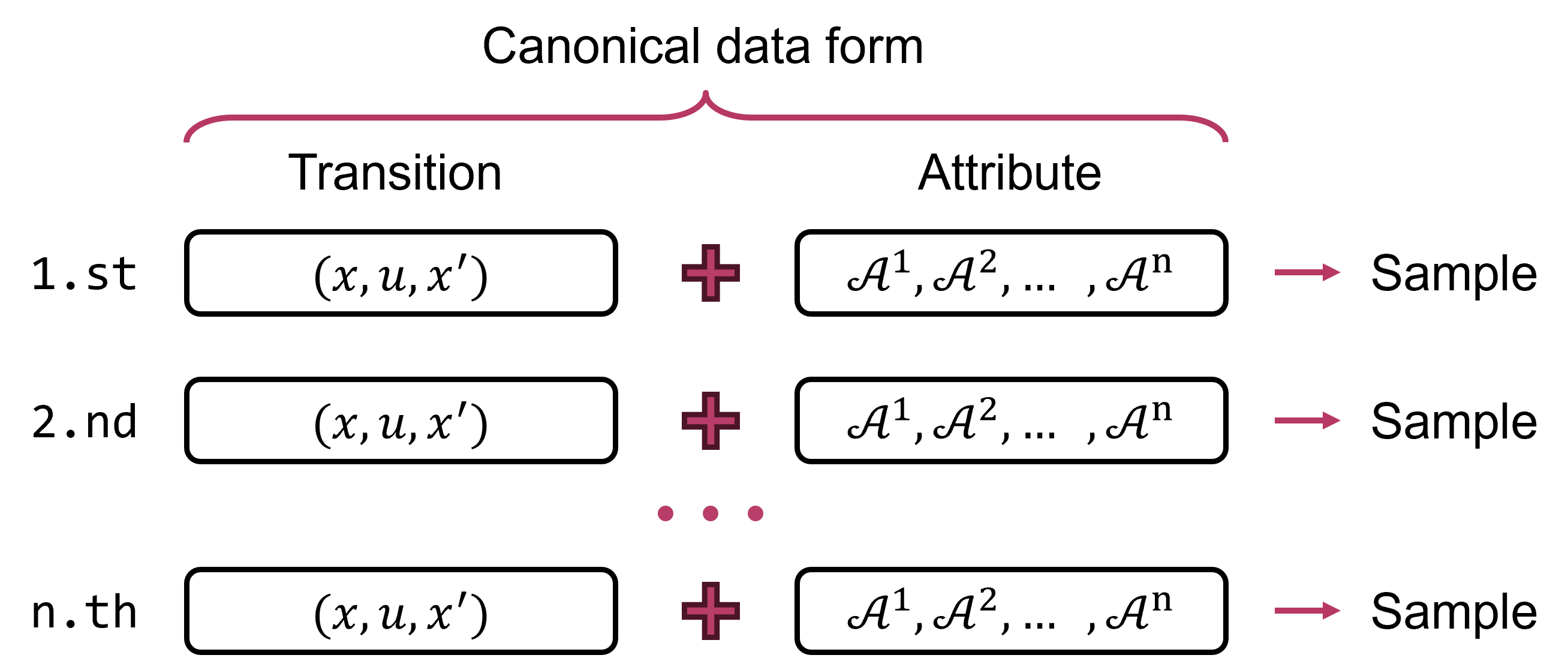}
\caption{Definition of canonical data form.}
\label{fig.datatic_canonical_form_framework}
\end{figure}

The canonical form in modelic control is a well-established concept,  
which involves utilizing a non-singular transformation matrix to convert the state-space model into an agreed-upon format. This form has two features: (1) it preserves the characteristics of plant dynamics ensured by the non-singular transformation; and (2) it introduces additional benefits for controller synthesis by customizing the agreed-upon format. For example, the controllable form defines the input matrix as a unit vector with the first element as 1. Additionally, it sets the subdiagonal elements of the system matrix to 1, allowing non-zero elements only in the last column. When a system model is described in the controllable form, its controllability can be easily checked without additional computation.
The Jordan form customizes the system matrix as a diagonal square matrix, where the diagonal elements indicate the system's eigenvalues. With this form, assessing stability is simplified by checking whether all eigenvalues are negative. Moreover, distinct eigenvalues indicate diverse system modalities, allowing for a more targeted controller design.

In datatic control, the canonical form is defined as an agreed-upon representation of data. To match its modelic counterpart, the canonical data form also have two features. One is to keep the characteristics of plant dynamics. 
The other is to offer some kinds of benefits for
controller design in terms of effectiveness and efficiency.
Therefore, we define a data sample with canonical form as 
\begin{equation}
    \text{Sample } S = \text{Transition } \mathcal{T} + \text{Attribute } \mathcal{A}.
\end{equation}
Here, $\mathcal{T}$ serves as a repository for preserving the system information. Specifically, it contains three elements: a state, an action and their corresponding next state, i.e., $\mathcal{T} = (x, u, x')$, which encapsulates the plant dynamics $f$ at each sampling time independently. 
The notation $\mathcal{A}$
describes one or some artificial characteristics of the current sample, whose calculation must be performed in an online manner. The two elements of each sample are illustrated in Fig. \ref{fig.datatic_canonical_form_framework}.
The attribute of each sample must adhere to two requirements:
(1) \textit{causality}, ensuring independence from any future samples;
and (2) \textit{locality}, allowing dependence on historical samples
but constrained to a finite neighboring set. 

Consider a standard sampling process that generates a trajectory denoted as $\mathcal{\tau} = \{x_1, u_1, ..., x_{T-1}, u_{T-1}, x_{T}\}$. This raw trajectory does not adhere to the template of canonical data form, and its canonical standardization can be performed at each sample in an online manner. Specifically, when the three variables $x_i, u_i, x_{i+1}$ are all known, the transition of the $i$-th sample can be easily built as $\mathcal{T}_i = (x_i, u_i, x_{i+1})$. 
Regarding its attribute, simultaneous computation is necessary during the sampling process. On one hand, causality must be satisfied, i.e., the attribute calculation must be independent of subsequent samples because future samples are unpredictable. On the other hand, locality must be satisfied, i.e., the attribute calculation should only be linked to a locally finite number of historical samples in order to avoid large computational burden. To give a practical example, the reward signal in model-free reinforcement learning is an instantaneous attribute that naturally meets the two requirements.

Subsequent sections will provide detailed illustrations of two representative canonical data forms: temporal form and spatial form, offering a comprehensive introduction to their content and benefits. 

\subsection{Temporal canonical form}
When controlling a system from one state to another, different policies will generate distinct state trajectories. The time consumed in this process is a critical indicator of controller performance and holds promise for enhancing controller design. Here, we define the temporal attribute as the time cost between pre-determined states, also termed as events, as illustrated in Fig. \ref{fig.temporal_canonical_form}.

\begin{figure}[!htbp]
\centering
\includegraphics[width=0.35\textwidth,trim={0cm 0.5cm 0cm 0.5cm}, clip]{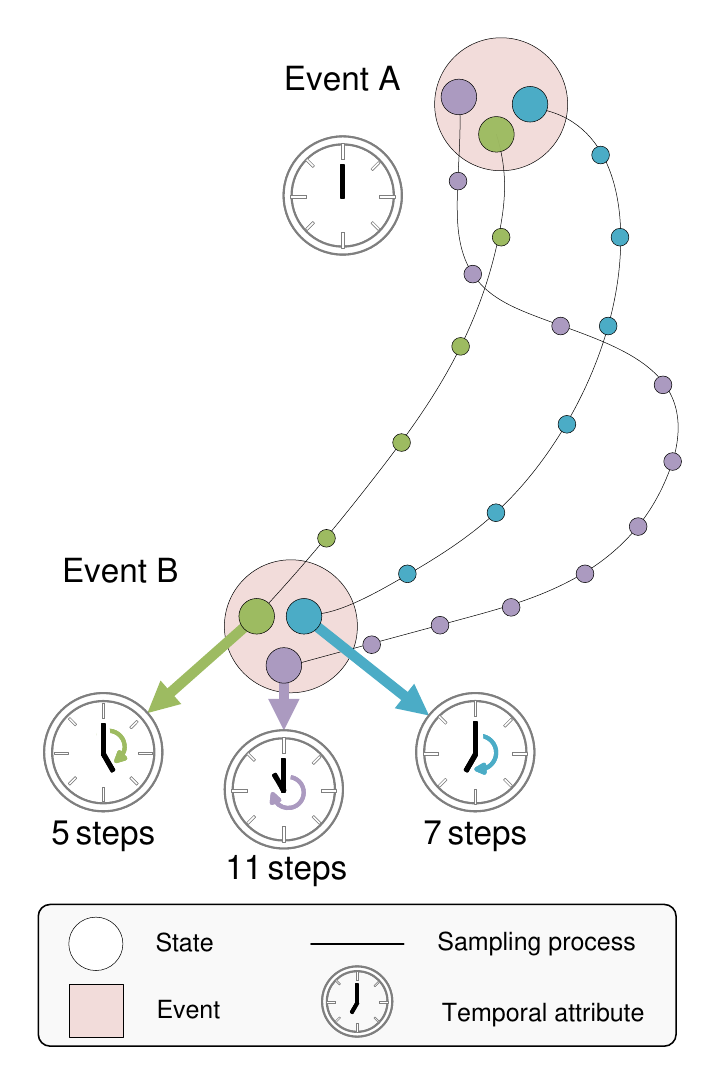}
\caption{Definition of temporal canonical data form.}
\label{fig.temporal_canonical_form}
\end{figure}

One can first identify the minimum time cost between two events and fit a piece-wise linear function between event and time, called event-time distribution. This distribution then serve as an additional source to guide better controller design.
For a more detailed illustration, let us consider a control task with two events as shown in Fig. \ref{fig.tem_2d}. There are two raw trajectories in blue and green colors generated in the sampling process, with two events highlighted in red bars. The blue trajectory triggers Event A twice in the first and third steps, and then triggers Event B at the sixth step. Therefore, for its sixth sample, its transition is denoted by a blue rectangle, and its temporal attribute is equal to 3 (steps), representing the minimum time cost from Event A to Event B along this trajectory. Similarly, the green trajectory triggers Event A at the second step and then triggers Event B twice at the fourth and seventh steps. Therefore, for its fourth sample, its temporal attribute is equal to 2 (steps) and its transition is denoted by a green rectangle.
After obtaining these event samples, the minimum time cost for each event pair is determined by identifying the lowest one. During the training process, when any event pair is triggered, the time cost of policy is compared with the fitted event-time distribution. And positive reward will be assigned if the policy spends less time, while negative punishments will be applied otherwise. This supplementary performance indicator is expected to alleviate instability issues in training a datatic controller, especially when dealing with sparse reward tasks.

\begin{figure}[!htbp]
\centering
{\includegraphics[width=0.4\textwidth]{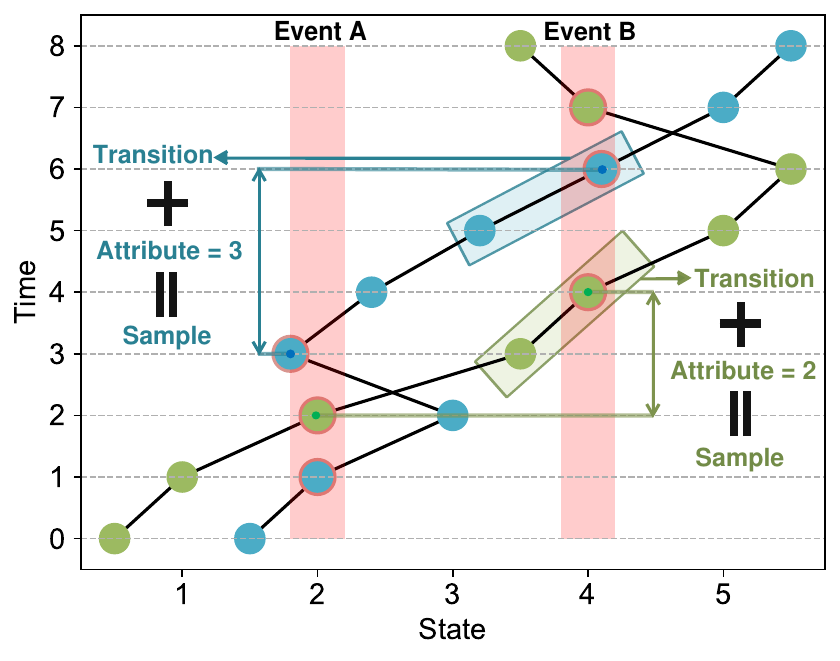}}
\caption{
Illustration of temporal canonical data form.}
\label{fig.tem_2d}
\end{figure}

\subsection{Spatial canonical form}

The samples generated by different policies have different degrees of spatial similarity. In many cases, neighboring samples around the selected one are useful for datatic controller design. Here, we set the spatial attribute as the distances to several pre-determined anchors as shown in Fig. \ref{fig.spatial_canonical_form}. This spatial attribute is helpful to accelerate the searching process for neighboring samples.

\begin{figure}[!htbp]
\centering
\includegraphics[width=0.35\textwidth,trim={0.0cm 0.5cm 0.0cm 0.5cm}, clip]{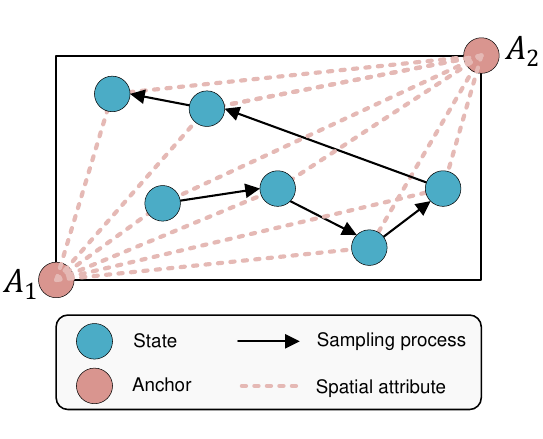}
\caption{Definition of spatial canonical data form.}
\label{fig.spatial_canonical_form}
\end{figure}

A representative searching task is the R-neighbor search illustrated in Fig. \ref{fig.anchor}. For a selected sample $C$ shown in green point, its R-neighbor area is a circle with a radius of $R$, shown in light green. The searching goal is to identify all samples within the R-neighbor area.
A trivial solution for this task is to traverse all the samples to calculate their distances to the selected sample, followed by a selection of those within the R-neighbor area.
Unlike this solution, we can utilize the saved spatial attribute to construct a filter condition,  whose purpose is to quickly reject most incompetent samples at first glance to accelerate the searching process. The mathematical principle underlying this filter condition is the triangle inequality, as demonstrated in the lemma below.
\begin{lemma}[Triangle inequality] The absolute difference between the lengths of any two sides of a triangle must be less than the length of the remaining third side. Mathematically, for a triangle with sides of lengths 
$a$, $b$, and 
$c$, they satisfy
\begin{equation}
    \Vert a - b\Vert  \leq c,  \Vert a - c\Vert  \leq b, \Vert b - c\Vert  \leq a.
\end{equation}
\label{triangle_inequality}
\end{lemma}
Following this lemma, we can derive the spatial filter condition outlined below.
\begin{theorem}[Spatial filter condition]
Consider a dataset with $n$ anchors denoted as $A_1, ..., A_n$, let $C$ be a selected sample and $S$ be any other sample. A necessary condition for $S$ in the R-neighbor area of $C$, termed as spatial filter condition, is given by
    \begin{equation}
        \big(\Vert A_1 S-A_1 C\Vert\leq R\big) \ \land, ..., \land \ \big(\Vert A_n S-A_n C\Vert \leq R\big),
    \label{eq.spatial_filter_condition}
    \end{equation}
    where $\land$ denotes the logical AND operator.
\end{theorem}
\begin{proof}
    For any anchor $A_i$, let $A_i S$, $A_i C$ and $SC$ form a triangle. According to Lemma \ref{triangle_inequality}, we have
    \begin{equation}
        SC \leq \Vert A_i S-A_i C\Vert.
    \end{equation}
    
    Then the condition $\Vert A_i S - A_i C\Vert \leq R$ is a necessary condition for   
    $SC \leq R$, which indicates the sample point $S$ lies in the R-neighbor area of $C$. Leveraging every anchor $A_i, i=1, ..., n.$, the spatial filter condition  \eqref{eq.spatial_filter_condition} is constructed.
\end{proof}
Applying this condition allows us to efficiently filter out a substantial number of samples that are surely not within the R-neighbor area, resulting in significant time savings.

\begin{figure}[!htbp]
\centering
\includegraphics[width=0.4\textwidth,trim={0cm 0.0cm 0cm 0.0cm}, clip]{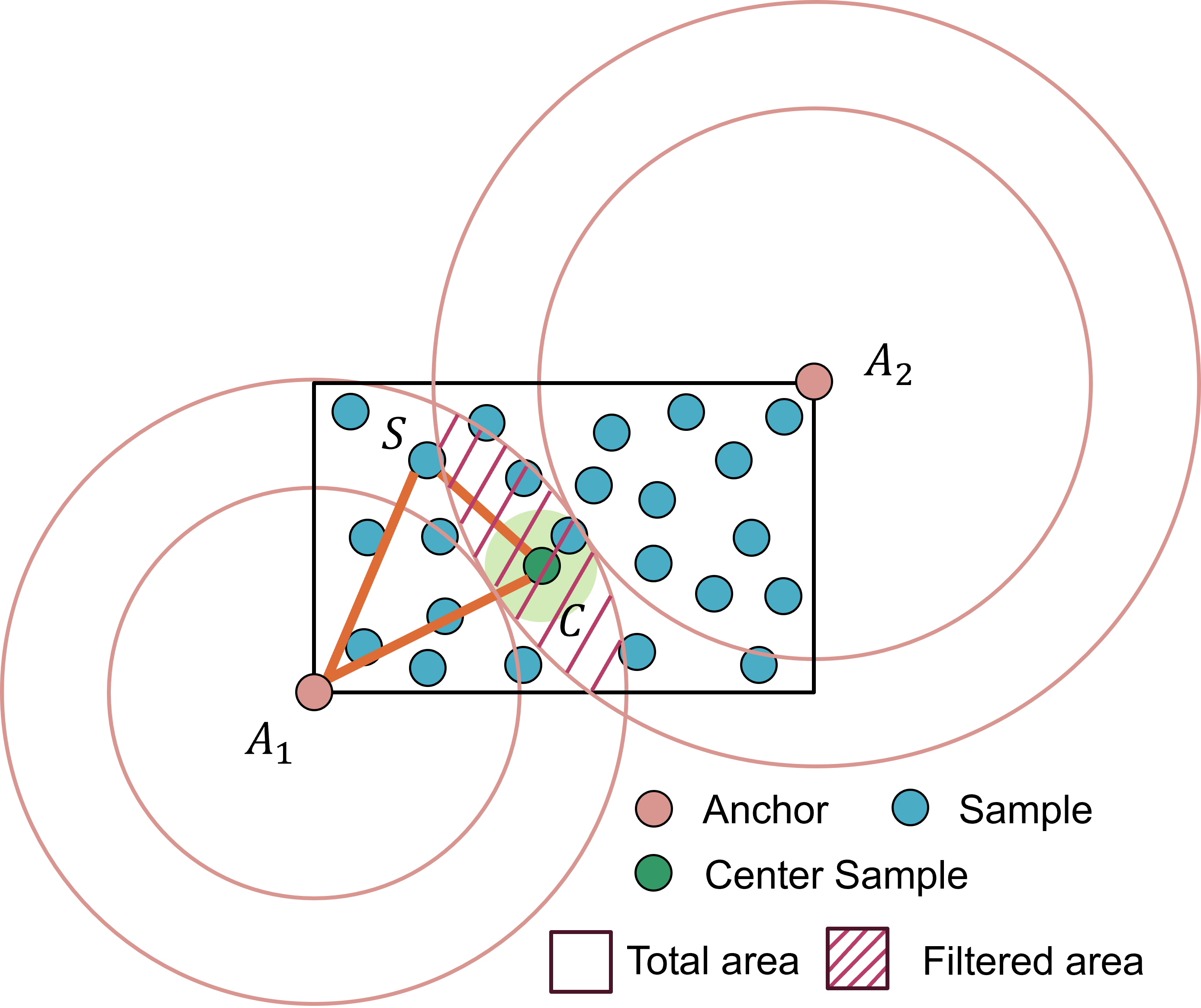}
\caption{Benefit of spatial canonical data form in R-neighbor search task. By applying the spatial filter condition constructed with two anchors $A_1$ and $A_2$, The filtered area represented by the shaded area is much smaller than the original total area.}
\label{fig.anchor}
\end{figure}
\section{Application}
This section validates the effectiveness of temporal and spatial canonical data by conducting experiments in two datatic control systems. One is controlling an underpowered car to reach the top of a mountain, and the other is controlling an one-legged hopper to move forward.

\subsection{Temporal canonical form}
\subsubsection{Task}
We select the MountainCar environment from Gym, as illustrated in Fig. \ref{fig.MountainCar task}, where an underpowered car is positioned between two mountains. The state $x \in \mathbb{R}^2$ comprises the car's horizontal position and velocity. The action $u \in \mathbb{R}$ is discrete, with alternatives of $[-1, 0, 1]$, representing left force, no force, and right force to drive the car, respectively.
The objective is to control the car to reach the top of the right mountain, with the evaluation metric being the time spent to achieve this target. The primary challenge arises from the underpowered nature of the car, preventing it from directly driving toward the destination. Specifically, the applied driving force is insufficient to overcome the climbing resistance force, necessitating a swinging motion between the two mountains until enough momentum is accumulated.

\begin{figure}[!htbp]
\centering
\includegraphics[width=0.35\textwidth,trim={0cm 2.0cm 0cm 1.4cm}, clip]{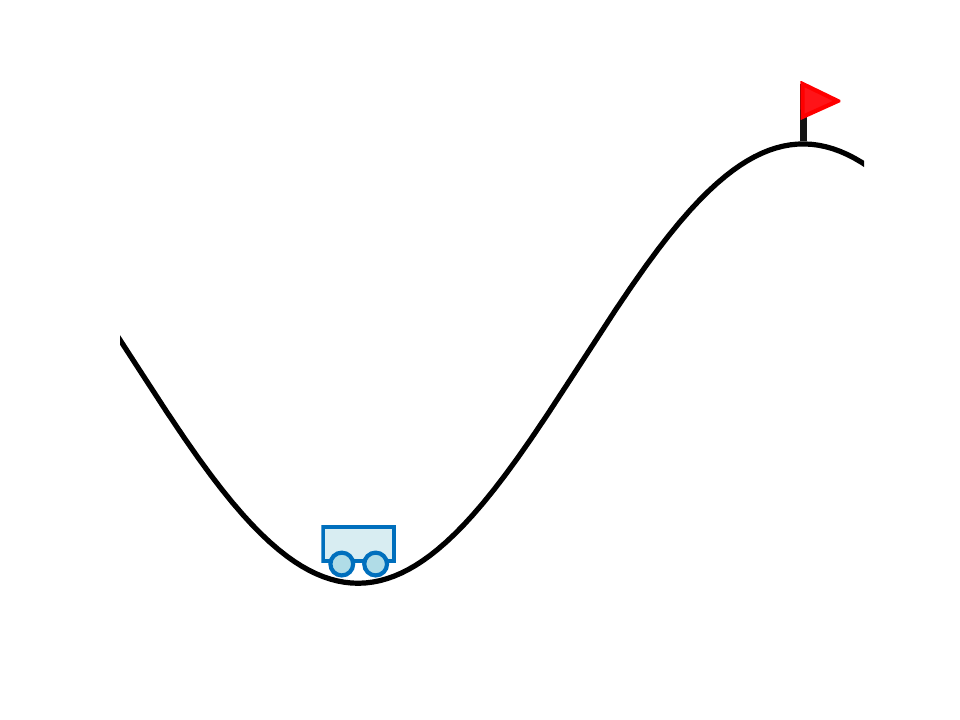}
\caption{MountainCar environment. The horizontal position of the initial bottom and the target flag at the right mountain are -0.5 and 0.5, respectively.}
\label{fig.MountainCar task}
\end{figure}

\subsubsection{Temporal attribute}
We define two events and record the time between them as the temporal attribute. The first event represents the starting situation, with the car initialized at the bottom, having a position of -0.5, and a velocity of 0. The second event corresponds to the halt situation, where the car comes to a stop due to insufficient momentum.
In each sampling process, the car is consistently initialized at the same starting bottom point. If any event corresponding to the halt situation is triggered, the time spent from the initial state will be recorded as the temporal attribute of the current sample. An illustrative example is presented in Fig. \ref{fig.temporal_attribute_in_sampling}.

\begin{figure}[!htbp]
\centering
\includegraphics[width=0.49\textwidth,trim={0.4cm 0.5cm 0.4cm 0.0cm}, clip]{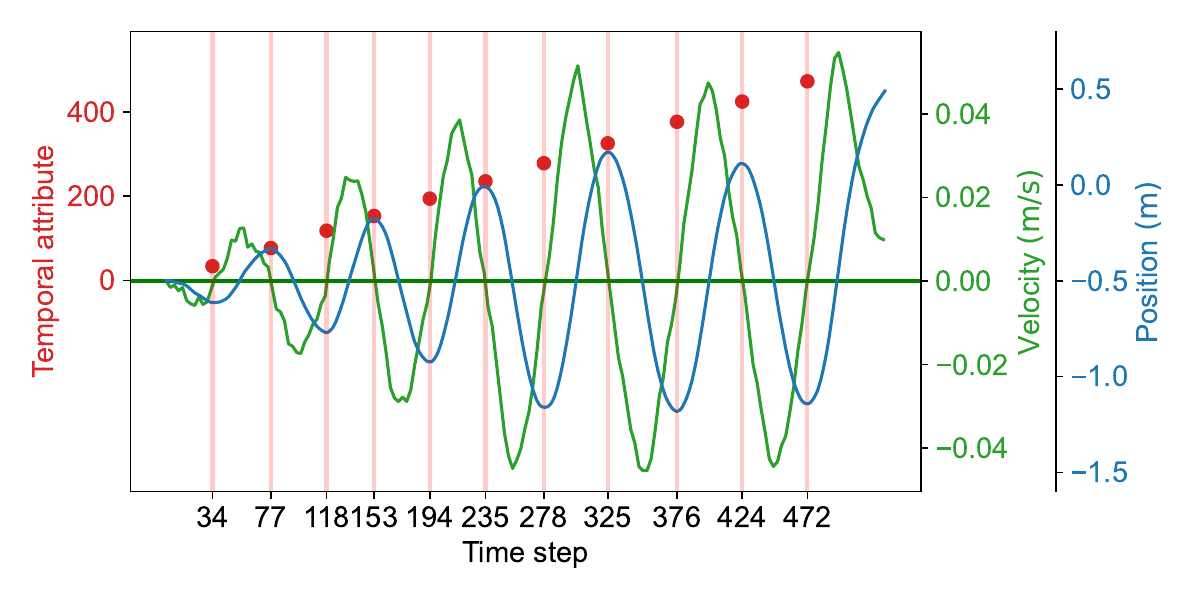}
\caption{Illustration of the temporal attribute in a specific sampling process.}
\label{fig.temporal_attribute_in_sampling}
\end{figure}

Utilizing the data described in temporal canonical form, we fit an event-time distribution, which represents the minimum time between the occurrence of the two events defined before. For the sampling process illustrated in Fig. \ref{fig.temporal_attribute_in_sampling}, there are 11 samples with temporal attributes. All of their velocities are near-zero, and the corresponding positions and temporal attributes are [-0.61, -0.33, -0.77, -0.18, -0.92, -0.01, -1.16, 0.17, -1.18, 0.11, -1.14], and [34, 77, 118, 153, 194, 235, 278, 325, 376, 424, 472], respectively. As shown in Fig. \ref{fig.temporal_fit0}, the horizontal axis is the position, and the vertical axis is the temporal attribute. A piece-wise linear function is then used to fit the inferior convex hull of these event samples, depicted as blue points. This fitted function serves as the event-time distribution, depicted by the green band.

\begin{figure}[!htbp]
\centering
\includegraphics[width=0.4\textwidth,trim={0.0cm 0.5cm 0.0cm 0.0cm}, clip]{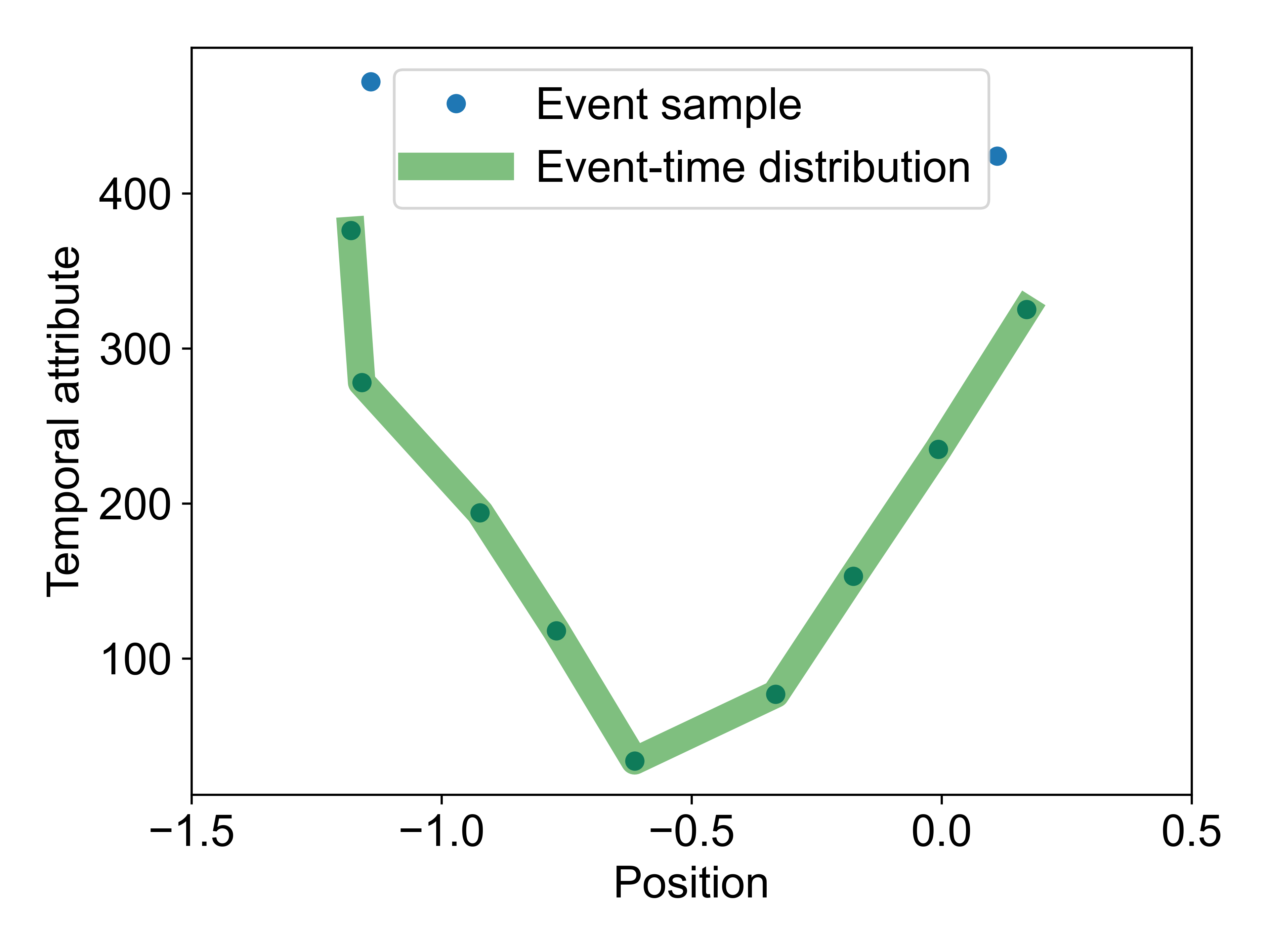}
\caption{An example of the fitted event-time distribution by the samples from a specific sampling process.}
\label{fig.temporal_fit0}
\end{figure}

\subsubsection{Algorithm}
We use the standard Deep Q Network (DQN) algorithm as the backbone \cite{wang2023gops}, and equip it with the fitted event-time distribution to serve as an additional performance measure. During each episode, when the car is initialized, the first event is triggered. Whenever the car subsequently triggers the second event, we compare the time consumed by the policy with the queried time cost from the event-time distribution. The difference between these values is then utilized as an extra reward signal, which is added into the original reward.

\subsubsection{Results}

Fig. \ref{fig.temporal_fit} illustrates the fitted event-time distribution of the collected samples in  replay buffer after 10000
iterations of training. The event samples with temporal attribute are represented as blue points, and the fitted inferior convex hull of these points is the event-time distribution, marked as the green band. This distribution indicates that the time cost increases as the car moves farther away from the bottom. 

\begin{figure}[!htbp]
\centering
\includegraphics[width=0.4\textwidth,trim={0.0cm 0.0cm 0.0cm 0.0cm}, clip]{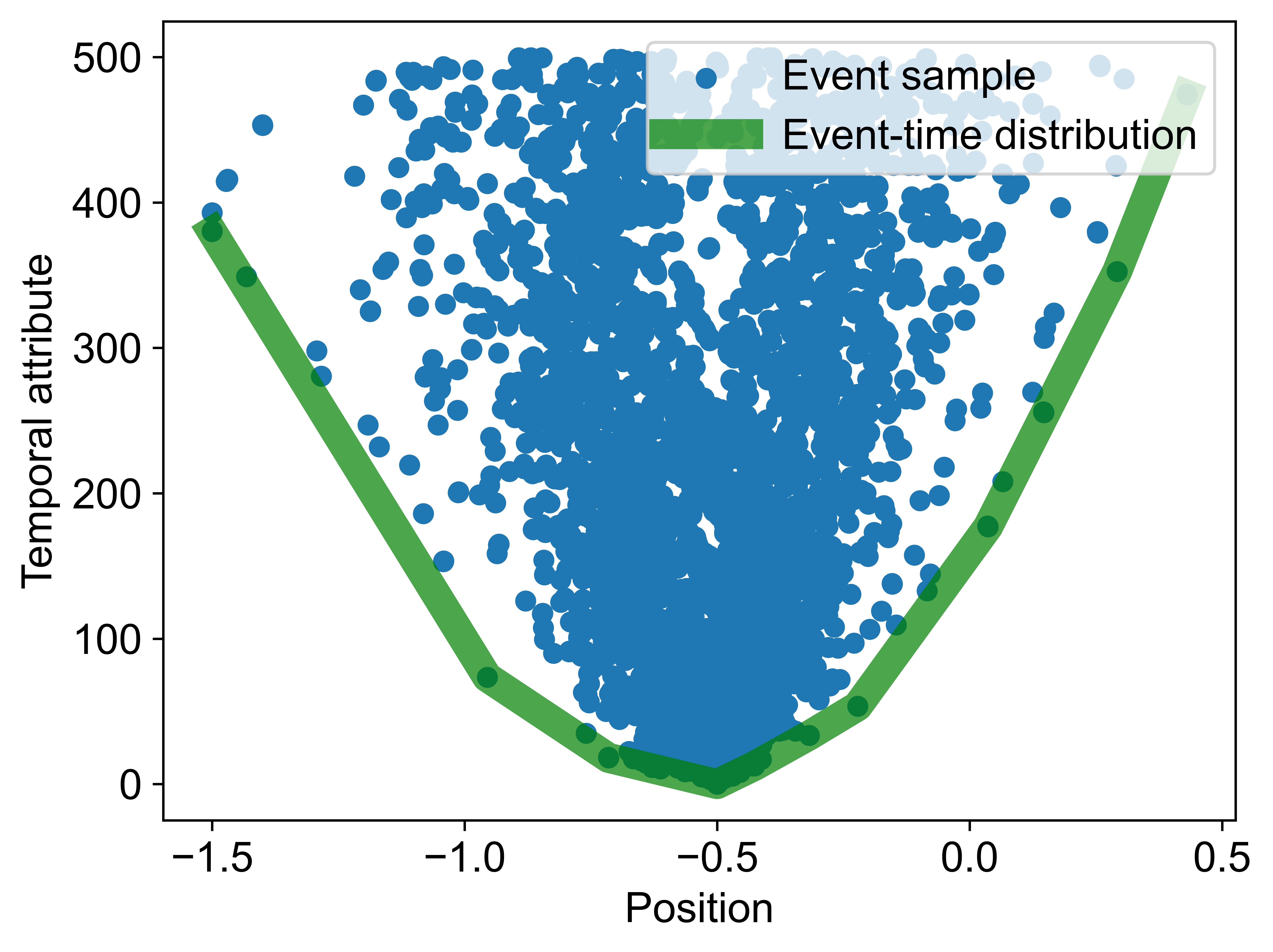}
\caption{The fitted event-time distribution for samples in the replay buffer after 10,000 iterations of training.}
\label{fig.temporal_fit}
\end{figure}

\begin{figure}[!htbp]
\centering
\hspace{25pt}
\includegraphics[width=0.4\textwidth,trim={0cm 0.8cm 0cm 0.8cm}, clip]{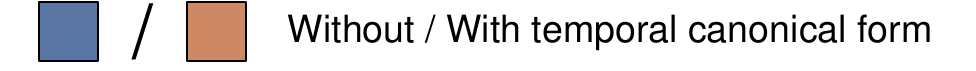}
\includegraphics[width=0.45\textwidth,trim={0cm 0.5cm 0cm 0.4cm}, clip]{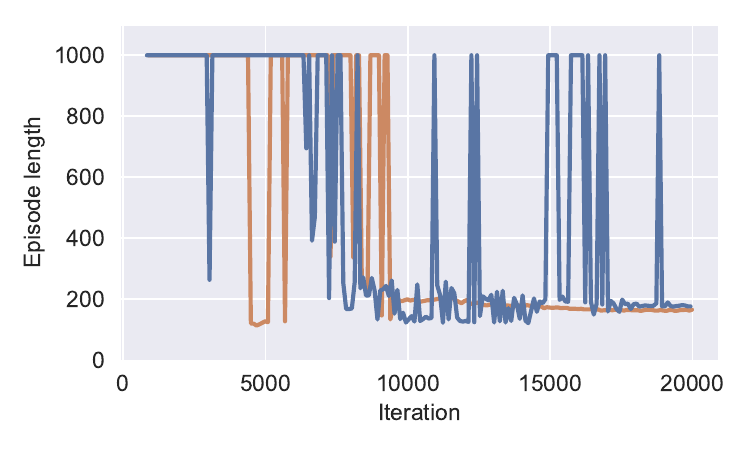}
\caption{Performance comparison during the training with/without temporal canonical data form. Here the vertical axis represents the episode length, denoting the time spent to achieve the target. We have set a maximum step limit of 1000, and the episodes whose length exceed this limit will be terminated.}
\label{fig.temporal_training}
\end{figure}

In Fig. \ref{fig.temporal_training}, the training processes across 20,000 iterations of two DQN methods are depicted. One incorporates the temporal canonical form, while the other does not. The blue curve represents the naive DQN without the temporal canonical form, displaying pronounced oscillations and inconsistent performance. In contrast, the orange curve corresponds to the DQN equipped with the temporal canonical form. Despite initial oscillations, it gradually comes to stabilization, achieving faster convergence and a modest performance improvement. This suggests that the temporal canonical form effectively mitigates instability issues, making the training process more efficient and reliable.

\subsection{Spatial canonical form}
\subsubsection{Task}
We choose the hopper environment in the MuJoCo locomotion benchmarks. The hopper is a two-dimensional one-legged robot that consist of four main body parts: the torso at the top, the thigh in the middle, the leg in the bottom, and a single foot on which the entire body rests. The goal is to make the hopper move in the forward (right) direction by applying torques on the three hinges connecting the four body parts. 
Specifically, the state $x\in\mathbb{R}^{11}$ consists of positional values of different body parts, followed by the velocities of those individual parts (their derivatives) with all the positions ordered before all the velocities, and the action $u\in\mathbb{R}^3$ is the torques applied on 3 hinges. The dataset we use is the ``hopper-medium-replay'' from the open-source D4RL repository. The selected dataset consists of 401,598 samples obtained from the replay buffer after 1 million iteration steps. The training algorithm is the Twin Delayed Deep Deterministic Policy Gradient (TD3) method.

\begin{figure}[!htbp]
\centering
\includegraphics[width=0.2\textwidth,trim={0cm 0cm 0cm 0cm}, clip]{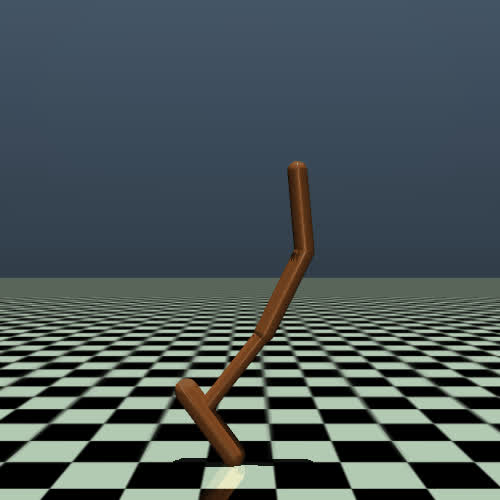}
\caption{Hopper task with 11 state dimensions and 3 action dimensions.}
\label{fig.hopper task}
\end{figure}

\subsubsection{Spatial attribute}

We first denote the unit vector as $e_i=[0, ..., 1, ..., 0]^\T$ whose all elements are 0 except the $i$-th element as 1. We consider both the state and action to define the feature distance as $(\beta x) \oplus u$, where $\oplus$ denotes the concatenation operation and $\beta$ is a hyper-parameter for trading off the impacts of state and action. Given that the sum of state and action dimensions is 14, we design 29 anchors, including an all-zero origin and the endpoints of the 14 axis of coordinates as
\begin{equation}
\textbf{0} \text{ and }
    \bigg[ e_i\times j \text{ for } i \in \{1, ..., 14\}  \text{ for }  j \in \{-1, 1\} \bigg].
\end{equation}

In each sampling process, the distances to these 29 anchors are recorded as the spatial attribute of the current sample. An illustrative example is presented in Fig. \ref{fig.spatial_attribute_in_sampling}. We select the 2nd state as the x-axis and the 6th state as the y-axis for a 2D projection. In this plane, only 5 anchors are left, which are  depicted as red circles. The sampled trajectory consists of a sequence of state denoted as points, with each point's color indicating the corresponding time step. The distances to all anchors, denoted as red dotted lines, are set as the spatial attribute of each sample.

\begin{figure}[!htbp]
\centering
\includegraphics[width=0.49\textwidth,trim={1.4cm 0.4cm 0.6cm 0.2cm}, clip]{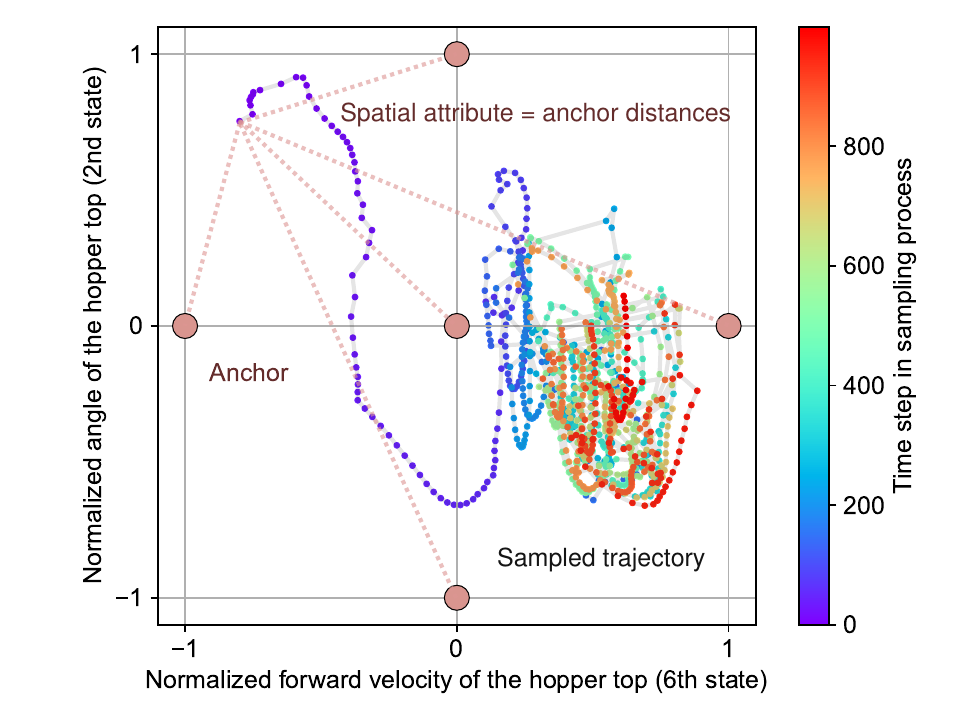}
\caption{Illustration of the spatial attribute in a specific sampling process.}
\label{fig.spatial_attribute_in_sampling}
\end{figure}

\subsubsection{Algorithm}
We employ the TD3 algorithm as the backbone, and enhance its resistance to out-of-distribution issue by incorporating dataset constraints. Specifically, for an arbitrarily selected sample with a transition denoted as $(x_c,u_c,x'_c)$, its R-Neighbor set $\mathcal{N} \subset \mathcal{D}$ is defined as
\begin{equation}
\begin{aligned}
    \mathcal{N} =\{(x,u,x')|
    \|(\beta x)\oplus u-(\beta x_c)\oplus u_c\| \leq R \},
    \label{eq.find_R_neighbor_set}
\end{aligned}
\end{equation}
where $R$ is the radius of the R-Neighbor area. And the point-to-dataset distance of the selected sample is given by
\begin{equation}
        d^\beta_\mathcal{N} = \text{mean}_{(x,u,x')\in\mathcal{N}}\|(\beta x)\oplus u-(\beta x_c)\oplus u_c\|.
\label{eq.point_to_set_distance}
\end{equation}

Based on the definition, we give the following dataset constraint loss as
\begin{equation}
\label{eq:dc_loss}
    \mathcal{L}_\text{DC}(\theta) = d^\beta_\mathcal{N}\left(x_c, \pi_\theta(x_c), x'\right),
\end{equation}
where $\theta$ denotes the learnable parameters of the policy $\pi_\theta$. 

Combining dataset constraint 
 in \eqref{eq:dc_loss} and the standard policy loss of  
 TD3, we derive the following total policy loss
\begin{equation}\label{eq:total_policy_loss}
    \mathcal{L}_{\text{Total}}(\theta) = \lambda \mathcal{L}_{\text{TD3}}(\theta) + \mathcal{L}_{\text{DC}}(\theta),
\end{equation}
where $\lambda$ is dynamically adjusted to balance maximizing  rewards and imitating dataset behaviors \cite{td3_bc}. 
Regarding the hyper-parameters related to the dataset constraint, we set $\beta$ as $0.2$ and $R$ as $0.5$. 

At each iteration, this algorithm needs to address a R-neighbor search task to calculate the dataset constraint loss. Therefore, utilizing the spatial filter condition provided by the spatial canonical form is expected to significantly enhance training efficiency, as it can filter out a substantial number of samples that are clearly not in the R-neighbor set.

\subsubsection{Results}
After 1000 iterations, the training processes of two TD3 algorithms are illustrated in Fig. \ref{fig.hopper_training_curve}, depicting their total average returns. One algorithm incorporates the spatial canonical form, while the other does not. The horizontal axis represents the wall time, indicating the real-world time spent. Notably, the utilization of spatial canonical form has significantly reduced the training time from over 20 hours to approximately 7 hours, resulting in a roughly threefold increase in training efficiency.

\begin{figure}[!htbp]
\centering
\includegraphics[width=0.48\textwidth,trim={0cm 0.4cm 0.cm 0.2cm}, clip]{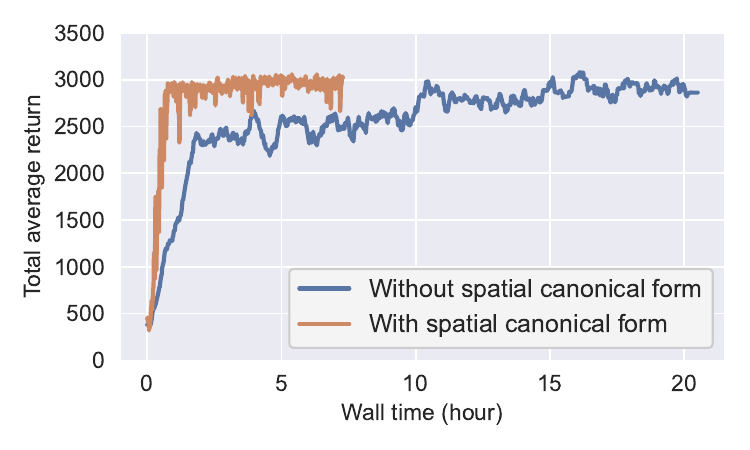}
\caption{Performance comparison during training process with/without spatial canonical form. Here, the vertical axis represents the total average return, calculated following the standard Mujoco benchmark. The horizontal axis indicates the wall time, i.e., the time consumed in the real-world. }
\label{fig.hopper_training_curve}
\end{figure}

To delve into the impact of spatial filter from a microscopic perspective, we randomly select 1000 samples from the training process and focus on the computation time comparison for R-Neighbor search tasks. 
We consider two algorithms: one is with our spatial canonical form and the other does not.
As depicted in Fig. \ref{fig.spatial_time}, the former,  with the help of spatial canonical form, reduces the computation time of each step to around $1 \mathrm{ms}$, while the latter incurs a much higher cost of approximately $20 \mathrm{ms}$.
Furthermore, we provide a visualization about the proportion of rejected data after applying the spatial filter as shown in Fig. \ref{fig.mask}. We define a metric named as reject ratio to evaluate the effective of spatial filter as
\begin{equation}
    \eta = 1 - \frac{N_\text{candidate}}{N_\text{total}},
\end{equation} 
where $N_\text{candidate}$ is the number of samples that satisfy \eqref{eq.spatial_filter_condition} and $N$ is the number of total samples. Higher reject ratio indicates better improvement in time efficiency.

We conduct 1000 tests using the randomly selected samples. It can be seen that for over 650 tests, the reject ratio is lager than $99.9\%$, and for over 950 tests, the reject ratio is lager than $99.5\%$. In other words, the spatial filter condition effectively reject a large proportion of samples.

\begin{figure}[!htbp]
\centering
\includegraphics[width=0.4\textwidth,trim={0cm 0.9cm 0cm 0.6cm}, clip]{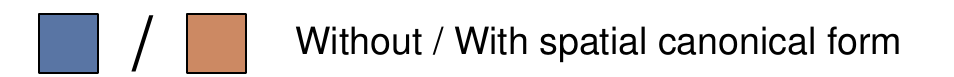}
\includegraphics[width=0.38\textwidth]{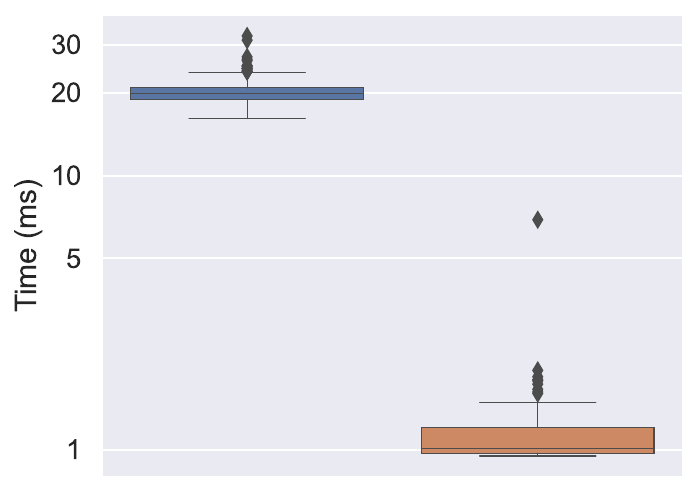}
\caption{The comparison of consumed time for accomplishing the R-neighbor search task is based on 1000 randomly selected samples, both without and with the spatial canonical form.}
\label{fig.spatial_time}
\end{figure}

\begin{figure}[!htbp]
\centering
\includegraphics[width=0.45\textwidth,trim={0cm 0.1cm 1.5cm 0.4cm}, clip]{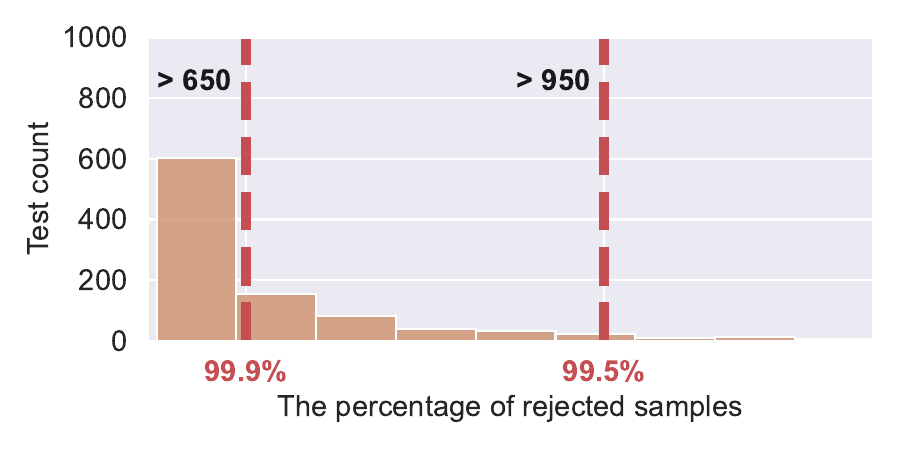}
\hspace{10pt}
\caption{Visualization of the effectiveness of spatial filter condition. The vertical axis is the test count, while the horizontal axis is the reject ratio, representing the percentage of samples rejected by the spatial filter.}
\label{fig.mask}
\end{figure}

\section{Conclusion}

This paper for the first time introduces the concept of \textit{canonical data form} into datatic control systems. In a datatic control system, the data sample in canonical form consists of a transition component and an attribute component. The former encapsulates the plant dynamics at the sampling time
independently. The latter describes one or some artificial characteristics of the current sample, whose calculation must be performed in an online manner. The attribute of each sample must adhere to two conditions: (1) causality, ensuring independence from any future samples; and (2) locality, allowing dependence on historical samples but constrained to a finite neighboring set. In our framework, different canonical forms can be customized according to specific needs to facilitate the development of datatic controllers.
Two representative canonical data forms, namely temporal form and spatial form, are presented as illustrations. This paper also provides a comprehensive introduction to their content and benefits in reducing instability and enhancing the training efficiency of datatic controller design.

\bibliographystyle{IEEEtran}
\bibliography{reference.bib}

\end{document}